\documentclass[runningheads,a4paper]{llncs}

\usepackage[T1]{fontenc}
\usepackage{url}

\usepackage[english]{babel}
\usepackage{times,amsmath,amssymb,graphicx,subfigure, xcolor,environ,tabularx,lipsum,complexity,multirow}
\usepackage{graphics}
\usepackage{setspace}

  \usepackage{tikz}
  \usetikzlibrary{positioning,fit,calc,shapes,backgrounds}
  \usetikzlibrary{matrix}

\makeatletter
\newcommand{\problemtitle}[1]{\gdef\@problemtitle{#1}}
\newcommand{\probleminput}[1]{\gdef\@probleminput{#1}}
\newcommand{\problemquestion}[1]{\gdef\@problemquestion{#1}}
\NewEnviron{myproblem}{
  \problemtitle{}\probleminput{}\problemquestion{}
  \BODY
  \par\addvspace{.5\baselineskip}
  \noindent
  \begin{tabularx}{\textwidth}{@{\hspace{\parindent}} l X c}
    \multicolumn{2}{@{\hspace{\parindent}}l}{\@problemtitle} \\
    \textbf{Input:} & \@probleminput \\
    \textbf{Question:} & \@problemquestion
  \end{tabularx}
  \par\addvspace{.5\baselineskip}
}

\newcommand{\pG}{\partial\Gamma}

\begin{document}

\mainmatter  


\title{b-continuity and Partial Grundy Coloring of graphs with large girth
\thanks{Partially supported by CNPq Projects Universal no. 401519/2016-3 and Produtividade no. 304576/2017-4, and by FUNCAP/CNPq project PRONEM no. PNE-0112-00061.01.00/16.}}
\titlerunning{b-continuity and Partial Grundy Coloring of graphs with large girth}
\authorrunning{Ibiapina and Silva}

\author{Allen Ibiapina\and Ana Silva}

\institute{ParGO Group - Paralelism, Graphs and Optimization\\Departamento de Matem\'atica\\Universidade Federal do Cear\'a\\
    Fortaleza, CE - Brazil\\
\url{allenr.roossim@gmail.com}; \url{anasilva@mat.ufc.br}}

\maketitle

\begin{abstract} 
  A b-coloring of a graph is a proper coloring such that each color class has at least one vertex which is adjacent to each other color class. The b-spectrum of $G$ is the set $S_{b}(G)$ of integers $k$ such that $G$ has a b-coloring with $k$ colors and $b(G)=\max S_{b}(G)$ is the b-chromatic number of $G$. A graph is b-continous if $S_{b}(G)=[\chi(G),b(G)]\cap \mathbb{Z}$. An infinite number of graphs that are not b-continuous is known. It is also known that graphs with girth at least 10 are b-continuous. 
  A partial Grundy coloring is a proper coloring $f:V(G)\rightarrow \{1,\ldots,k\}$ such that each color class $i$ contains some vertex $u$ that is adjacent to every color class $j$ such that $j<i$. The partial Grundy number of $G$ is the maximum value $\pG(G)$ for which $G$ has a partial Grundy coloring. 
  
  In this work, we prove that graphs with girth at least 8 are b-continuous, and that the b-spectrum of a graph $G$ with girth at least 7 contains the integers between $2\chi(G)$ and $b(G)$. We also prove that $\pG(G)$ equals a known upper bound when $G$ is a graph with girth at least~7. These results generalize previous ones by Linhares-Sales and Silva (2017), and by Shi et al.(2005).
\end{abstract}

\section{Introduction} 

Let $G$ be a simple graph (for basic terminology on graph theory, we refer the reader to~\cite{BM08}). A function $\psi \colon V(G)\to \mathbb{N}$ is a \emph{proper k-coloring of G} if $|\psi(V(G))| = k$ and $\psi(u)\neq \psi(v)$ whenever $uv \in E(G)$. Because we only deal with proper colorings in this text, from call them simply a coloring. We call the elements of $\psi(V(G))$ \emph{colors}. Given a color $i \in \psi(V(G))$, the set $\psi^{-1}(i)$ is called \emph{color class i}. The \emph{chromatic number of $G$} is the minimum value $\chi(G)$ for which $G$ admits a coloring. The problem of deciding~$\chi(G)\le k$ for given $G$ and $k$ is one of the most studied problems in Graph Theory. It is one of Karp's 21 $\NP$-complete problems~\cite{K.72}, and it remains $\NP$-complete even when restricted to many classes of graphs, as for instance, it is $\NP$-complete when $k=3$ and $G$ is a 4-regular planar graph~\cite{D.80}, or a line graph~\cite{H.81}. There are many variations of the coloring problem, and in this article we investigate two of these variations, both of them based on existing coloring heuristics.

Given a coloring $\psi$ of $G$, we say that $u \in V(G)$ is a \emph{b-vertex in $\psi$} if $u$ is adjacent to every color class different from $\psi(u)$, i.e., if $N(u)\cap \psi^{-1}(i)\neq \emptyset$ for every $i\neq \psi(u)$. Note that if some color class $c$ does not contain b-vertices, we can obtain a $(k-1)$-coloring by changing the color of each vertex $v \in \psi^{-1}(c)$ to another color different from $c$ where $v$ has no neighbors. We say that this new coloring is obtained from the first one by \emph{cleaning  color $c$}. If a coloring cannot be cleaned, it means that all color classes have at least one b-vertex. Such a coloring is called  a \emph{b-coloring of $G$}. Observe that an optimal coloring cannot have the number of colors decreased by the described algorithm; therefore every optimal coloring is also a b-coloring. In~\cite{Irving.Manlove.99}, the authors define the \emph{b-chromatic number of $G$}, denoted by $b(G)$, as the largest natural $k$ for which $G$ has a b-coloring with $k$ colors. In the same article, the authors demonstrated that the problem of finding  $b(G)$ is $\NP$-complete in general.

Another interesting aspect about b-colorings concerns its existence for every possible value between $\chi(G)$ and $b(G)$. In~\cite{Irving.Manlove.99}, the authors observe that the cube has a b-coloring using 2 colors and 4 colors, but has no b-coloring using 3 colors. Inspired by this result, in~\cite{KRATOCHVIL.etal.02} it is shown that for any integer $n\geq 4$ the graph obtained from the complete bipartite graph $K_{n,n}$ by deleting the edges from a perfect matching has a b-coloring using $2$ and $n$ colors, but has no b-coloring using a number of colors between $2$ and $n$. This motivates the definition of the \emph{b-spectrum of G}, that is the set $S_{b}(G)$ containing every integer $k$ such that $G$ has a b-coloring with $k$ colors. A graph $G$ is \emph{b-continous} if $S_{b}(G)= [\chi(G),b(G)] \cap \mathbb{Z}$. In~\cite{Barth.Cohen.Faik}, they prove that for each finite subset $S \subset \mathbb{N} - \{1\}$, there exists a graph $G$ such that $S_{b}(G)=S$, and also that deciding if a graph is b-continuous is $\NP$-complete even if colorings with $\chi(G)$ and $b(G)$ colors are given.

Now, given a b-coloring with $k$ colors, since each b-vertex has at least $k-1$ neighbors, there exists $k$ vertices with degree at least $k-1$ (this would be a subset of $k$ b-vertices of the $k$ colors). So if we define $m(G)$ as the largest positive integer $k$ such that there exist at least $k$ vertices with degree at least $m(G)-1$ in $G$, we have that $b(G)\leq m(G)$. 
This upper bound was introduced in~\cite{Irving.Manlove.99}, where the authors show that one can find $m(G)$ in polynomial time using the degree list of the graph. 
Also, they prove that if $G$ is a tree, then $b(G)\geq m(G)-1$, and that one can decide if $b(G)=m(G)$ in polynomial time. Their result was later generalized for graphs with girth at least 7~\cite{CLS.15} (the \emph{girth of $G$} is the minimum lenght of a cycle in $G$). We also mention that there are many results that say that regular graphs with large girth have high b-chromatic number~\cite{BMZ.09,CJ.11,EK.09,BMZ.09}. Indeed, the following conjecture is still open.

\begin{conjecture}[Blidia, Maffray and Zemir~\cite{BMZ.09}]
If $G$ is a $d$-regular graph with girth at least~5 and $G$ is not the Petersen graph, then $b(G)= d+1$.
\end{conjecture}

Because of these results, it makes sense to investigate the b-continuity of graphs with large girth. Indeed, in~\cite{BK.12} the authors prove that regular graphs with girth at least~6 and without cycles of length~7 are b-continuous, and in~\cite{L-SS.17}, they prove that every graph with girth at least~10 are b-continuous. 
Here, we improve their result to graphs with girth at least~8.

\begin{theorem}\label{th:main} If $G$ is a graph with girth at least 8, then $G$ is b-continuous.
\end{theorem}

In addition, we prove that graphs with girth at least~7 are, in a way, almost b-continuous.

\begin{theorem} \label{th:second}
If $G$ is a graph with girth at least 7, then $[2\chi(G),b(G)]\cap \mathbb{Z} \subseteq S_{b}(G)$.
\end{theorem}

Now, given a coloring $\psi:V(G)\rightarrow \{1,\ldots,k\}$, we say that $u\in V(G)$ is a \emph{Grundy vertex} if $u$ has a neighbor in color class $i$, for every color $i<f(u)$, and that $f$ is a \emph{partial Grundy coloring} if every color class contains at least one Grundy vertex. This concept was introduced in~\cite{EHLP.03}, where the authors relate the new concept with parsimonious colorings. We mention that a more largely studied concept is that of the Grundy colorings, where every vertex is a Grundy vertex. These are also called greedy colorings since they represent colorings that can be obtained by the greedy algorithm. In this context, a partial Grundy coloring can be thought of as a coloring that cannot have the number of colors decreased by moving vertices to smaller colors. Because a coloring with $\chi(G)$ colors cannot be improved like this, we get that it must be a partial Grundy coloring. 
The \emph{partial Grundy number} is then defined as the worst case scenario for this heuristis, i.e., it is the maximum number of colors in any partial Grundy coloring; it is denoted by $\pG(G)$.

Interestingly enough, these colorings are also easier on graphs with large girth. In~\cite{S.etal.05}, the authors introduce an upper bound for $\pG(G)$, called the \emph{stair factor of $G$} and denoted here by $s(G)$, and they prove that if $G$ has girth at least~9, then $\pG(G) = s(G)$. We improve their result in the theorem below. 

\begin{theorem} \label{th:partialGrundy}
If $G$ is a graph with girth at least 7, then $\pG(G)=s(G)$, and an optimal partial Grundy coloring can be obtained in polynomial time.
\end{theorem}

Observe that this and the result in~\cite{CLS.15} tell us that that if $G$ has girth at least~7, then both the b-chromatic number and the partial Grundy number of $G$ are high. However, we mention that these colorings are not similar when the continuity is concerned. In~\cite{BK.13}, the authors prove that every graph $G$ has a partial Grundy coloring with $k$ colors, for every $k\in\{\chi(G),\ldots,\pG(G)\}$. This is a common trait with Grundy colorings~\cite{CS.79}.

Our article is organized as follows. In Section~\ref{sec:proofBcol}, we prove Theorems~\ref{th:main} and~\ref{th:second}; in Section~\ref{sec:proofPartialG}, we prove Theorem~\ref{th:partialGrundy}; and in Section~\ref{sec:conclusion}, we give our closing remarks. The definitions and notation are introduced as they are needed.



\section{b-continuity}\label{sec:proofBcol}

In~\cite{BK.12}, a vertex $u\in V(G)$ is called a \emph{$k$-iris} if there exists $S \subseteq N(u)$ such that $|S| \geq k-1$ and $d(v)\geq k-1$ for every $v \in S$. See Figure~\ref{fig:kiris}. This definition is important because of the following important lemma.

\begin{figure}[thb]
\begin{center}
\begin{tikzpicture}
  \tikzstyle{level 1}=[sibling distance = 1cm]
  \tikzstyle{level 2}=[sibling distance = 0.5cm]
  [parent anchor= east, child anchor=west, grow=east]
  \tikzstyle{every node}=[draw, circle]
  \tikzstyle{edge from parent}=[draw]
  \node{$u$}
      	child {node {}
            child {node {}} child {node {}} }
	child {node {}
            child {node {}} child {node {}} }
	child {node {}
            child {node {}} child {node {}} };
    \end{tikzpicture}
   \caption{In the figure, we present a 4-iris.}
   \label{fig:kiris}
\end{center}
\end{figure}
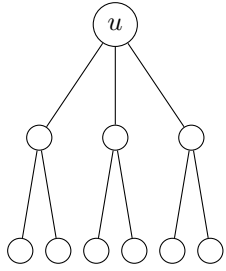

%
%
%
%
%
%
%

\begin{lemma}[\cite{BK.12}]\label{lem:BK.12}
Let $G$ be a graph with girth at least 6 and without cycles of lengh 7. If $G$ has a $k$-iris with	 $k \geq \chi(G)$, then $G$ has a b-coloring with $k$ colors.
\end{lemma}

To prove Theorems~\ref{th:main} and~\ref{th:second}, given a b-coloring of $G$ with $k$ colors, $k>\chi(G)+1$, we try to obtain a b-coloring of $G$ with $k-1$ colors. However, this is not always possible, and when this happens, it is because we have a $k$-iris. Our theorem then follows from the lemma above. We mention that the constraint about not having cycles of length~7 appears only in the above lemma, but not on our proof. We now introduce the further needed definitions.

Let $G$ be a simple graph and $\psi$ be a b-coloring of $G$ with $k>\chi(G)+1$ colors. We say that $u$ \emph{realizes color $i$} if $\psi(u)=i$ and $u$ is a b-vertex. We also say that color $i$ is \emph{realized by $u$}. For $x \in V(G)$ and $i \in \{1,\ldots,k\}$, let $N^{\psi,i}(x)$ be the set of vertices of color $i$ in the neighborhood of $x$, i.e., $N^{\psi,i}(x) = N(x)\cap \psi^{-1}(i)$; in fact, we omit $\psi$ in the superscript since it is always clear from the context. This is also done in the next definitions. For a subset $X \subseteq V(G)$, let $N^{i}(X) = \bigcup_{x \in X} N^{i}(X)\setminus X$. Let $B(\psi)$ denote the set of b-vertices in $\psi$ and, for each $i\in \{1,\ldots,k\}$, let $B_{i} = B(\psi)\cap \psi^{-1}(i)$ be the set of b-vertices in color class $i$. 

Given a set $K$ such that $K\subseteq \psi^{-1}(i)$ for some $i \in \{1,\ldots,k\}$, we say that a color $j \in \{1,\ldots,k\}\setminus\{i\}$ is \emph{dependent on K} if $N^{i}(B_{j}) \subseteq K$; denote by $U(K)$ the set of colors depending on $K$. If $K=\{x\}$, we write simply $U(x)$. 
Given $x \in V(G)\setminus B(\psi)$, if $|U(x)|\geq 2$ we call $x$ a \emph{useful} vertex; otherwise, we say that $x$ is \emph{useless}. 
For $j\in \{1,\ldots,k\}$, we say that $x \in V(G)$ is \emph{j-mutable} if $x$ is useless and there exists a color $c$ such that we can change the color of $x$ to $c$ without creating any b-vertex of color $j$; we also say that color $c$ is \emph{safe for $(x,j)$}. If there is no safe color for $(x,j)$, we say that $x$ is $j$-imutable.

The next lemma is the main ingredient in our proof. Combined with Lemma~\ref{lem:BK.12}, it immediatly implies Theorem~\ref{th:main}.

\begin{lemma}\label{lem:mainl}
Let $G=(V,E)$ be a graph with girth at least 7. If $G$ has a b-coloring with $k$ colors where $k\geq \chi(G)+1$, then either $G$ has a b-coloring with $k-1$ colors, or $G$ contains a $(k-1)$-iris.
\end{lemma}

\begin{proof}
Our proof is similar to that made in \cite{L-SS.17}, but we concentrate in one color that we want to eliminate. 

Suppose that $G$ does not have a b-coloring with $k-1$ colors; we prove that $G$ has a $(k-1)$-iris. 
For this, let $\psi$ be a b-coloring with $k$ colors that minimizes $|B_{1}|$ and then minimizes $|\psi^{-1}(1)|$ (i.e., it firstly minimizes the number of b-vertices of color~1, then it minimizes the number of vertices of color~1). First, we prove that every $x \in \psi^{-1}(1)\setminus B_1$ is useful. Suppose otherwise and let $x$ be a useless vertex in color class~1, i.e., $|U(x)|\le 1$. If $U(x)= \emptyset$, then we can recolor $x$ without losing any b-vertex, a contradiction since $\psi$ minimizes $|\psi^{-1}(1)|$. And if $U(x)=\{d\}$, then we can obtain a b-coloring with $k-1$ by recoloring $x$ and cleaning $d$, again a contradiction. Therefore, the following holds:	

\begin{center}(i) Every $x\in \psi^{-1}(1)\setminus B_1$ is useful. \end{center}

Now, we choose any $u \in B_{1}$ and analyse its vicinity in order to obtain the desired $(k-1)$-iris. For this, the following two claims are essential.

\begin{claim}\label{claim1}
Let $j\in \{2,\ldots,k\}$. If every $x\in N^j(u)\setminus B_j$ is 1-mutable, then one of the following holds:
\begin{itemize}
	\item[(ii)] $N(u)\cap B_{j} \neq \emptyset$; or

	\item[(iii)] There exists a color $d \in \{2,\ldots,k\}\setminus \{j\}$ such that $d$ depends on $N^{j}(u)$, i.e., 
	\[N^j(B_d)\subseteq N^j(u).\]
\end{itemize}
\end{claim}
\emph{Proof of claim:}
Suppose that neither (ii) nor (iii) holds, and let $\psi'$ be obtained from $\psi$ by changing the color of each $x\in N^{j}(u)$ to a color $c$ safe for $(x,1)$. Because $(iii)$ does not hold, we get that $U(N^{j}(u))\subseteq \{1\}$. Therefore, at most one color loses all of its b-vertices, namely color~1, and since every $x\in N^j(u)$ is 1-mutable, no b-vertices of color~1 is created. But because $u$ is not a b-vertex in $\psi'$ (it is not adjacent to color $j$ anymore) and $\psi$ minimizes $|B_{1}|$, we get that $\psi'$ cannot be a b-coloring, which means that we can obtain a b-coloring with $k-1$ colors by cleaning color 1.$\diamondsuit$

\vspace{0.3cm}
The following claim tells us that (ii) or (iii) actually always hold.

\begin{claim}\label{claim2} (iv) Every $x \in N^{j}(u)\setminus B_j$ is 1-mutable, for every $j\in \{2,\ldots,k\}$.\end{claim}
\emph{Proof of claim:}
Suppose, without loss of generality, that $d\in\{2,\ldots,k\}$ is such that the colors in $\{d+1,\ldots,k\}$ are exactly the colors that contains some $1$-imutable vertex. We count the number of colors with b-vertices in the vicinity of $u$ to get that in fact $d\ge k$. So, for each $i\in \{d+1,\ldots,k\}$, let $w_i\in N^i(u)$ be a $1$-imutable vertex. By definition, this means that, for each $i\in \{d+1,\ldots,k\}$, there exists some neighbor of $w_i$ that would be turned into a b-vertex of color~1 in case we change the color of $w_i$; let $v_i$ be such a vertex. We then know that $v_i\in \psi^{-1}(1)\setminus B_1$, which by (i) gets us that $|U(v_{i})|\geq 2$. By the definition of $U(x)$ and the fact that every $x\in \{v_{d+1},\ldots,v_k\}$ is colored with color~1, we get:

\begin{equation}U(v_i)\cap U(v_\ell) = \emptyset\mbox{, for every }i,\ell\in\{d+1,\ldots,k\}, i\neq \ell.\label{eq1}\end{equation}

Now, we investigate the b-vertices around colors $\{2,\ldots,d\}$. By Claim~\ref{claim1}, suppose, without loss of generality, that $p\in \{2,\ldots,d\}$ is such that (ii) holds for colors  in $\{2,\ldots,p\}$, while (iii) holds for colors in $\{p+1,\ldots, d\}$. For each $i\in \{p+1,\ldots, d\}$, let $c_i\in\{2,\ldots,k\}\setminus \{i\}$ be a color depending on $N^i(u)$, which means that $B_{c_i}\subseteq N(N^i(u))$.  Observe that, since $G$ has no cycles of length~3, we get:

\begin{equation}\{2,\ldots,p\}\cap \{c_{p+1},\ldots,c_d\} = \emptyset\label{eq2}\end{equation}

Also, because $G$ has no cycles of length~4, we get $c_i\neq c_\ell$ for every $i\neq \ell$, i.e.:

\begin{equation}|\{c_{p+1},\ldots,c_d\}| = d-p\label{eq3}\end{equation}

Finally, because $G$ has no cycles of length smaller than~6, we get that:

\begin{equation}\{2,\ldots,p,c_{p+1},\ldots,c_d\} \cap \bigcup_{i={d+1}}^{k}U(v_i) = \emptyset.\label{eq4}\end{equation}

Now, recall that $\psi(v_i)=1$ for every $i\in\{d+1,\ldots,k\}$, and that $c_i\neq 1$ for every $i\in\{p+1,\ldots,d\}$. This means that $1\notin \{c_{p+1},\ldots,c_d\}\cup\bigcup_{i=d+1}^pU(v_i)$. By combining Equations~(\ref{eq1}) through~(\ref{eq4}), we get the following, which implies $d\ge k$ as desired:

\[\begin{array}{ll}
k-1  & \geq |\{2,\ldots,p\}\cup \{c_{p+1},\ldots,c_{d}\} \cup U(v_{d+1})\cup \ldots \cup U(v_{k})| \\
              & = d- 1+ \sum_{i=d+1}^{k}|U(v_{i})| \\
              & \geq d -1 + 2(k-d).
\end{array}\diamondsuit\]

Now, let $N = (N(u)\cup N(N(u)))\setminus \{u\}$. Observe that because (ii) or (iii) holds for every color $\ell\in \{2,\ldots,k\}$, we get that $B(\psi)\subseteq N$. Suppose that $N[u]$ does not contain a $(k-1)$-iris, otherwise the proof is done. This means that at least one color in $\{2,\ldots,k\}$, say $k$, is such that (ii) does not hold for $k$, which  by Claim~\ref{claim1} implies that (iii) holds, i.e., that there exists a color in $\{2,\ldots,k-1\}$, say 2, such that $N^2(B_k)\subseteq N^2(u)$ (Observe Figure~\ref{fig:theoremProof}). Now, let $w\in N^1(B_k)$; it exists since the vertices in $B_k$ are b-vertices. By (i), there exists at least two colors in $\{2,\ldots,k\}$ that depend on $w$. But because $B(\psi)\subseteq N$, we get a cycle of length at most~6, a contradiction.\hfill$\Box$

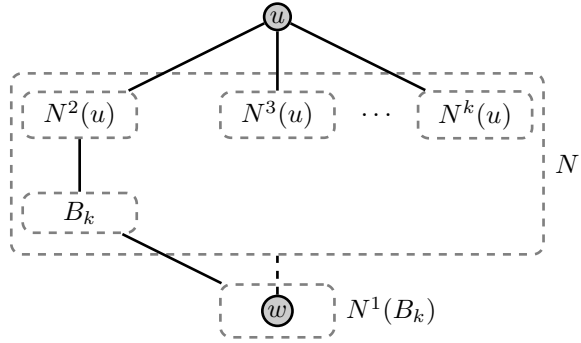
\begin{figure}[thb]
\begin{center}
  \begin{tikzpicture}[scale=1.3]
  \pgfsetlinewidth{1pt}
  \tikzset{vertex/.style={circle, minimum size=0.3cm, fill=black!20, draw, inner sep=1pt}}
  \tikzset{colorclass/.style ={draw=black!50, rounded corners, dashed, minimum height=5mm, minimum width=1.5cm} }

    \node [vertex] (u) at (0,0){$u$};
    \node [colorclass] (c2) at (-2,-1){$N^2(u)$};
    \node [colorclass] (Bk) at (-2,-2){$B_k$};
    
    \node [label=0:$N^1(B_k)$, draw=black!50, rounded corners, dashed, minimum height=7mm, minimum width=1.5cm] (N1Bk) at (0,-3){$$};
    \node [vertex] (w) at (0,-3) {$w$};
    
    \node [colorclass] (c3) at (0,-1){$N^3(u)$};
    \node [colorclass] (ck) at (2,-1){$N^k(u)$};
    \node at (1,-1) {$\ldots$};

    \node [label=0:$N$, draw=black!50, rounded corners, dashed, minimum height=24mm, minimum width=7cm] (N) at (0,-1.5){};
    
    \draw (u)--(c2) (u)--(c3) (u)--(ck) (c2)--(Bk)--(N1Bk);
    \draw[dashed] (w)--(N);

%

  \end{tikzpicture}
\caption{Structure around $u$ when color $k$ does not satisfy Claim~\ref{claim1}.(ii). }
\label{fig:theoremProof}
\end{center}
\end{figure}

\end{proof}

Now, to prove Theorem~\ref{th:second}, we apply Lemma~\ref{lem:mainl} and the next lemma. A \emph{star} is a tree that has at most one vertex with degree bigger than~1, and the \emph{diameter} of a graph $G$ is the maximum number of edges in a shortest path of $G$.

\begin{lemma}
If a graph $G$ has girth at least 7 and a $k$-iris where $k\geq 2\chi(G)$, then $G$ has a b-coloring with $k$ colors.
\end{lemma}

\begin{proof}
Let $u \in V(G)$ be a $k$-iris with $k\geq 2\chi(G)$. Let $u_{2},\ldots,u_{k}$ be neighbors of $u$ such that $d(u_{i})\geq k-1$ for every $i \in\{2,\ldots,k\}$; let $N_i$ be a subset of $k-2$ neighbors of $u_i$ different from $u$.  Start by coloring $u$ with 1 and, for each $i\in \{2,\ldots,k\}$, give color $i$ to $u_i$ and colors $\{2,\ldots,k\}\setminus \{i\}$ to $N_i$. Denote by $P$ the set $\bigcup_{i=2}^kN[u_i]$ (set of colored vertices), and observe that the coloring can be easily done since $G[T]-u$ is a forest formed by $k-1$ stars. Also, note that we already have $k$ b-vertices of distinct colors, and thus it only remains to extend the partial coloring to the rest of the graph. For this, let $B$ be the set of vertices adjacent to some color class in $\{1,\ldots,\chi(G)\}$. See Figure~\ref{fig:kirisG7}.

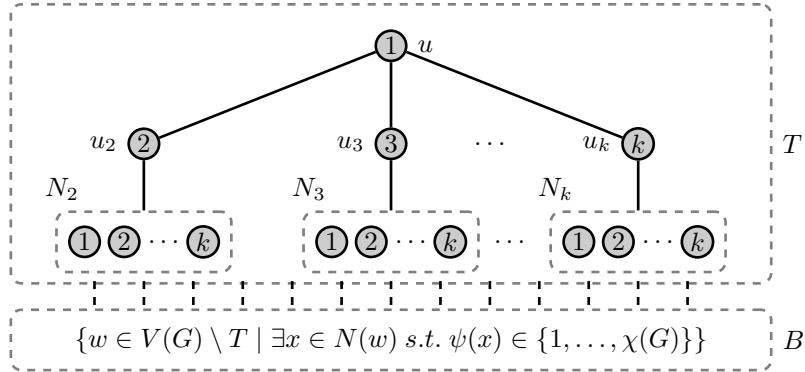
\begin{figure}[thb]
\begin{center}
  \begin{tikzpicture}[scale=1.3]
  \pgfsetlinewidth{1pt}
  \tikzset{vertex/.style={circle, minimum size=0.4cm, fill=black!20, draw, inner sep=1pt}}
  \tikzset{colorclass/.style ={draw=black!50, rounded corners, dashed, minimum height=8mm, minimum width=2.3cm} }

    \node [vertex,label=0:$u$] (u) at (0,0){$1$};
    \node [vertex,label=180:$u_2$] (u2) at (-2.5,-1){$2$};
    \node [vertex,label=180:$u_3$] (u3) at (0,-1){$3$};
        \node at (1,-1) {$\ldots$};
    \node [vertex,label=180:$u_k$] (uk) at (2.5,-1){$k$};
    
    \node [colorclass,label=150:$N_2$] (N2) at (-2.5,-2){$$};
    \node [colorclass,label=150:$N_3$] (N3) at (0,-2){$$};
        \node at (1.2,-2) {$\ldots$};
     \node [colorclass,label=150:$N_k$] (Nk) at (2.5,-2){$$};
     
     \node [vertex] at (-3.1,-2) {$1$};
     \node [vertex] at (-2.7,-2) {$2$};
     \node at (-2.3,-2) {$\ldots$};
     \node [vertex] at (-1.9,-2) {$k$};

     \node [vertex] at (-0.6,-2) {$1$};
     \node [vertex] at (-0.2,-2) {$2$};
     \node at (0.2,-2) {$\ldots$};
     \node [vertex] at (0.6,-2) {$k$};
     
     \node [vertex] at (1.9,-2) {$1$};
     \node [vertex] at (2.3,-2) {$2$};
     \node at (2.7,-2) {$\ldots$};
     \node [vertex] at (3.1,-2) {$k$};

    \node [label=0:$T$, draw=black!50, rounded corners, dashed, minimum height=37mm, minimum width=10cm] (N) at (0,-1){};
    
    \node [label=0:$B$, draw=black!50, rounded corners, dashed, minimum height=8mm, minimum width=10cm] (B) at (0,-3){$\{w\in V(G)\setminus T\mid \exists x\in N(w)\ s.t.\ \psi(x) \in\{1,\ldots,\chi(G)\}\}$};
    
    \draw (u)--(u2) (u)--(u3) (u)--(uk) (u2)--(N2) (u3)--(N3) (uk)--(Nk);
        \draw[dashed] (-2.5,-2.4)--(-2.5,-2.7) (-3,-2.4)--(-3,-2.7) (-2,-2.4)--(-2,-2.7) (-1.5,-2.4)--(-1.5,-2.7) (-1,-2.4)--(-1,-2.7) (-0.5,-2.4)--(-0.5,-2.7) (0,-2.4)--(0,-2.7) (2.5,-2.4)--(2.5,-2.7) (3,-2.4)--(3,-2.7) (2,-2.4)--(2,-2.7) (1.5,-2.4)--(1.5,-2.7) (1,-2.4)--(1,-2.7) (0.5,-2.4)--(0.5,-2.7);

%

  \end{tikzpicture}
\caption{Subset of vertices around $u$. The label inside a vertex denotes its color.}
\label{fig:kirisG7}
\end{center}
\end{figure}

We claim that $|N(w)\cap T|\le 1$ for every $w\in B$; indeed, because $T$ induces a tree of diameter~4,  if this was not true than we would get a cycle of length at most~6. By the definition of $B$, we then get that every $w\in B$ has no neighbors in color classes $\chi(G)+1,\ldots,k$. Thus, since $k\ge 2\chi(G)$, we can color $G[B]$ with colors $\chi(G)+1,\ldots,2\chi(G)$. Finally, by the definition of $B$, we know that every $w\in V(G)\setminus (B\cup T)$ has no neighbors of color $1$ through $\chi(G)$, which means that we can color $G-T-B$ with these colors.\hfill$\Box$
\end{proof}

\section{Partial Grundy Colorings}\label{sec:proofPartialG}

The next definitions were introduced in~\cite{S.etal.05}. Given $W = (w_1,\ldots,w_s)$, a sequence of vertices, we say that $W$ is a \emph{feasible sequence} if $d_{G_i}(w_i)\ge i-1$, for every $i\in \{1,\ldots,s\}$, where $G_i = G-\{w_{i+1},\ldots,w_s\}$. The \emph{stair factor of $G$} is the maximum size of a feasible sequence; we denote it here by $s(G)$. In~\cite{S.etal.05}, they prove that $s(G)$ is an upper bound for $\pG(G)$, and that if $W$ is a feasible sequence of $G$ of size $s$, and $G$ has girth at least~9, then $G$ has a partial Grundy coloring with $s$ colors. This implies that $\pG(G) = s(G)$ when $G$ has girth at least~9. We prove that it continues to hold when $G$ has girth at least~7 by also proving that the existence of a feasible sequence of size $s$ makes it possible to obtain a partial Grundy coloring of $G$ of size $s$. We mention that they also prove that computing a feasible sequence of maximum size can be done in polynomial time, which combined with our proof imply Theorem~\ref{th:partialGrundy}.

Now, consider a feasible sequence $(w_1,\ldots,w_s)$, $s\ge 2$, and for each $i\in \{2,\ldots,s\}$, let $N_i \subseteq N_{G_i}(w_i)$ be of cardinality $i-1$. Denote the set $\{w_2,\ldots,w_s\}$ by $W$. We say that a partial coloring $\psi$ of $G$ is \emph{non-redundant for $W$} if for every $i\in \{2,\ldots,s\}$, we get that $\psi(w_i) = i$ and all the colored vertices in $N_i$ have distinct colors in $\{1,\ldots,i-1\}$. 

\begin{proof}[of Theorem~\ref{th:partialGrundy}]
Let $(w_1,\ldots,w_s)$ be a feasible sequence. Denote by $W$ the set $\{w_2,\ldots,w_s\}$ and let $N_i$ be defined as before; also let $N = \bigcup_{i=2}^s N_i$. First, note that if $\psi$ is a non-redundant partial coloring that colors $W\cup N$, then a partial Grundy coloring with at least~$s$ colors can be obtained from $\psi$ in a greedy way: this is true because $w_i\in W$ is a Grundy vertex in $\psi^{-1}(i)$, for each $i\in \{2,\ldots,s\}$, and any vertex $x\in \psi^{-1}(1)$, which exists since $N_2$ is colored with~1, is a Grundy vertex. 

Start by coloring each $w_i\in W$ with $i$. For each $x\in N$, define:
\[m(x) = \min\{i\in\{2,\ldots,s\}\mid w_i\in N(x)\}.\]

Note that if $\psi$ is a non-redundant partial coloring and $x$ is colored in $\psi$, then $\psi(x)<m(x)$. We will then color $N$ in a non-decreasing order of $m(x)$. Suppose we are at the iteration of some $x\in N$, and observe that the colors that are forbidden for $x$ are the colors that appear in $N(x)$, and also the colors that appear in $N(w_i)$, where $w_i\in N(x)$. So, denote by $\ell$ the value $m(x)$ and by $F$ the set of vertices $N(x)\cup N(N_W(x))$. By what was said before, if there exists a color $c\in \{1,\ldots,\ell-1\}\setminus \psi(F)$, then we can give color $c$ to $x$ to obtain a non-redundant partial coloring that colors $x$, so suppose otherwise. For each $i\in \{1,\ldots,\ell-1\}$, let $y_i\in F$ be such that $\psi(y_i) = i$. First, note that, by the definition of $m(x)$ and the fact that $\psi(y_i)<m(x)$, we get that if $y_i\in N(x)$, then $y_i\in N$. Now, for each $i\in \{1,\ldots, \ell-1\}$, define: 

\[
u_i = \left\{\begin{array}{ll}
y_i & \mbox{, if $y_i\in W$}\\
w_{m(y_i)} & \mbox{, otherwise}
\end{array}\right.\]

 Let $W^{<\ell}$ be the set $\{w_2,\ldots,w_{\ell-1}\}$.  Observe Figure~\ref{fig:partialG}.

\begin{figure}[thb]
\begin{center}
  \begin{tikzpicture}[scale=1.3]
  \pgfsetlinewidth{1pt}
  \tikzset{vertex/.style={circle, minimum size=0.4cm, fill=black!20, draw, inner sep=1pt}}
  \tikzset{black/.style={circle, minimum size=0.4cm, fill=black, draw, inner sep=1pt}}
    \tikzset{u/.style={rectangle, minimum size=0.4cm, fill=black, draw, inner sep=1pt}}

    \node [vertex,label=150:$x$] (x) at (0,0){};
    \node [vertex,label=210:$y_{i_1}$] (yi1) at (-4,-1){}; 
    \node [u,label=270:$w_{m(y_{i_1})}$] (wi1) at (-4,-2){}; 
    \node [vertex,label=210:$y_{i_q}$] (yiq) at (-2,-1){};
    \node [u,label=270:$w_{m(y_{i_q})}$] (wiq) at (-2,-2){};
    \node at (-3,-1.5) {$\ldots$};
    
   \node [black] (t1) at (-1,-1){};     
    \node [u,label=270:$y_{i_{q+1}}$] (z1) at (-1,-2){}; 
    \node [black] (t2) at (1,-1){};    
   \node [u,label=270:$y_{i_r}$] (z2) at (1,-2){};
       \node at (-0,-1.5) {$\ldots$};

    \node [black] (wqp1) at (2,-1){}; 
    \node [vertex,label=-30:$y_{i_{r+1}}$] (yqp1) at (2,-2){}; 
    \node [u,label=270:$w_{m(y_{i_{r+1}})}$] (wmqp1) at (2,-3){}; 
    \node [black] (wl) at (4,-1){}; 
    \node [vertex,label=210:$y_{i_\ell}$] (yl) at (4,-2){}; 
    \node [u,label=270:$w_{m(y_{i_\ell})}$] (wml) at (4,-3){}; 
    \node at (3,-1.5) {$\ldots$};

   \draw (wi1)--(yi1)--(x)--(yiq)--(wiq) (wml)--(yl)--(wl)--(x)--(wqp1)--(yqp1)--(wmqp1) (z2)--(t2)--(x)--(t1)--(z1);
  \end{tikzpicture}
\caption{Subset of vertices in the vicinity of $x$. Black vertices are in $W$. Square vertices are in $\{u_1,\ldots,u_\ell\}$. Observe that no two of them are equal since $G$ has girth at least~7.}
\label{fig:partialG}
\end{center}
\end{figure}
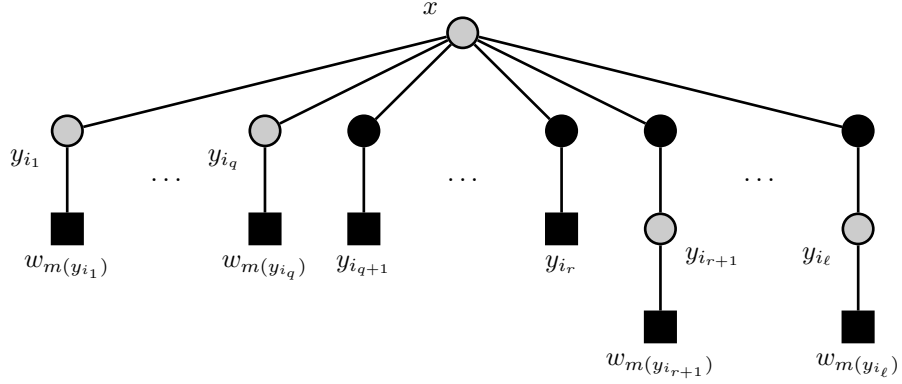

 Consider $i\in \{1,\ldots,\ell-1\}$. If $u_i=y_i$, then, because $\psi(y_i)<m(x)$, we get $u_i\in W^{<\ell}$.  Otherwise, we get that $y_i\in N$ and, by the chosen order of coloring, we get that $m(y_i) < m(x)$, which again implies that $u_i\in W^{<\ell}$. We then get $\{u_1,\ldots,u_{\ell-1}\}\subseteq W^{<\ell} = \{w_2,\ldots,w_{\ell-1}\}$, which by the pigeonhole principle implies that there exist $i,j\in \{1,\ldots,\ell-1\}$ with $i\neq j$ such that $u_i=u_j$. But because $u_i$ and $u_j$ are at distance at most~3 from $x$, we get a cycle of length smaller than~7, a contradiction.\hfill$\Box$
\end{proof}

\section{Conclusion}\label{sec:conclusion}

We have proved that every graph with girth at least~8 is b-continuous, and that graphs with girth at least~7 are in a way almost b-continuous. This improves the result presented in~\cite{L-SS.17}, where they prove that graphs with girth at least~10 are b-continuous. There, the authors also pose the following questions:

\begin{question}\label{q:general}
What is the minimum $\hat{g}$ such that $G$ is b-continuous whenever $G$ is a graph with girth at least $\hat{g}$?
\end{question}

\begin{question}\label{q:bipartite}
What is the minimum $\tilde{g}$ such that $G$ is b-continuous whenever $G$ is a bipartite graph with girth at least $\tilde{g}$?
\end{question}

Recall that the graph obtained from the complete bipartite graph $K_{n,n}$ by removing a perfect matching is not b-continuous, for every $n\ge 4$~\cite{KRATOCHVIL.etal.02}. Hence, by our result we get:

\[5\le \hat{g}\le 8.\]

And because bipartite graphs have no odd cycles, we get:
\[\tilde{g}\in\{6,8\}.\]

We believe that the same techniques might improve this bound to~7, but not further. In particular, we mention that Lemma~\ref{lem:mainl} works for graphs with girth~7, and that Theorem~\ref{th:main} is restricted to girth at least~8 just because of Lemma~\ref{lem:BK.12}. Therefore, if the following question is answered ``yes'', then we get $\hat{g}\le 7$.

\begin{question}
Let $G$ be a graph with girth at least~7 such that $G$ has a $k$-iris, with $k \geq \chi(G)+1$. Does $G$ admit a b-coloring with $k$ colors?
\end{question}

We also mention that the coloring problem is $\NP$-complete for graphs with girth at least $k$, for every fixed $k\geq 3$~\cite{LK.07}. This is why any proof of a result like Theorem~\ref{th:main} is expected to have a non-constructive part. 

As for the case of bipartite graphs, we think it is worth mentioning a known conjecture about their b-chromatic number. Recall the upper bound $m(G)$ for the b-chromatic number $b(G)$, which is the maximum value $k$ for which there exist $k$ vertices with degree at least $k-1$. The set of all vertices with degree at least $m(G)-1$ is denoted by $D(G)$, and a graph is said to be \emph{tight} if $|D(G)| = m(G)$; this means that there is only one candidate set for the set of b-vertices of a b-coloring of $G$ with $m(G)$ colors. Deciding if $b(G)=m(G)$ is $\NP$-complete even for bipartite tight graphs~\cite{KRATOCHVIL.etal.02}. In~\cite{HLS.12}, the authors define the class $\mathcal{B}_m$ that contains every bipartite graph $G = (A\cup B, E)$ such that $m(G)=m$, $D(G) = A$ and $G$ has girth at least~6. They conjecture the following:

\begin{conjecture}[Havet, Linhares-Sales and Sampaio~\cite{HLS.12}\label{ConjHSL.12}]
For every $m\ge 3$, and every $G\in \mathcal{B}_m$, we have that: \[b(G)\ge m(G)-1.\]
\end{conjecture}

We mention that, if $G$ is a bipartite graph with girth at least~6 and a b-coloring of $G$ with $k$ colors is given, with $k\ge \chi(G)+1$, then, with a little more work in the proof of Lemma~\ref{lem:mainl}, one can get that $G$ contains an induced subgraph $H$ that has a structure similar to the structure of a graph in $\mathcal{B}_k$. Trying to use this structure to obtain a b-coloring of $H$ with $k-1$ colors could translate into proving Conjecture~\ref{ConjHSL.12}. And on the other way around, we believe that a strategy to prove Conjecture~\ref{ConjHSL.12} could help coloring these graphs, which would imply that the answer to Question~\ref{q:bipartite} is ``yes''. This means that answering Question~\ref{q:bipartite} seems as hard as proving Conjecture~\ref{ConjHSL.12}. We also mention that in~\cite{LC.13}, it is proved that Conjecture~\ref{ConjHSL.12} is a consequence of the famous Erd\H{o}s-Faber-Lov\'asz Conjecture, which remains open since~1972 and which is largely believed to hold. This is strong evidence that Conjecture~\ref{ConjHSL.12} holds.

%

Now, concerning partial Grundy colorings, an obvious question also concerns the minimum girth for which $\pG(G)=s(G)$ holds. 

\begin{question}\label{q:general}
What is the minimum $\dot{g}$ such that $\pG(G)=s(G)$ whenever $G$ is a graph with girth at least $\dot{g}$?
\end{question}

Oberve that if $G=(A\cup B,E)$ is the complete bipartite graph $K_{n,n}$, then $s(G) = n+1$: a feasible sequence is obtained by starting with any vertex in $A$, then ordering each $v\in B$. However, we argument that the only existing partial Grundy coloring of $G$ uses 2 colors. For this, suppose without loss of generality that color~1 appears in $A$. Then no vertex in $B$ is colored with 1, which means that no Grundy vertex of color bigger than~1 can appear in $A$. But now, if color $i>2$ has a Grundy vertex $u$ in $B$, then color~2 must appear in $A$ and, therefore, cannot appear in $B$. We get a contradiction since in this case color~2 has no Grundy vertices. We then get the following bounds:
\[5\le \dot{g}\le 7.\]

\bibliographystyle{plain}
\bibliography{refs}


%
%
%
%
%
 


\end{document}